\documentclass[11pt]{article}

\usepackage{amsthm}
\usepackage{algorithm, algorithmic}
\usepackage{graphicx}
\usepackage{amssymb}
\newtheorem{thm}{Theorem}[section]
\theoremstyle{definition}
\newtheorem{dfn}{Definition}[section]
\theoremstyle{remark}

\theoremstyle{plain}
\newtheorem{lem}[thm]{Lemma}
\newtheorem{col}[thm]{Corollary}

\newtheorem{fig}[figure]{Fig.}

\def\denseformat{
\setlength{\textheight}{9.2in}
\setlength{\textwidth}{6.9in}
\setlength{\evensidemargin}{-0.2in}
\setlength{\oddsidemargin}{-0.2in}
\setlength{\headsep}{10pt}
\setlength{\topmargin}{-0.3in}
\setlength{\columnsep}{0.375in}
\setlength{\itemsep}{0pt}
}
\def\MathN{\hbox{\rm I\kern-2pt I\kern-3.1pt N}}
\denseformat
\begin{document}
\title{Arboricity}

\title{Deterministic Distributed Vertex Coloring in Polylogarithmic Time
}
\author{
Leonid Barenboim\thanks{
        Department of Computer Science,
        Ben-Gurion University of the Negev,
        POB 653, Beer-Sheva 84105, Israel.
        E-mail: {\tt \{leonidba,elkinm\}@cs.bgu.ac.il}
        \newline This research has been supported by the
Israeli Academy of Science, grant 483/06, and by the Binational Science Foundation, grant No. 2008390.
        }
         \and
Michael Elkin$^*$ }
\begin{titlepage}
\def\thepage{}
\maketitle
\begin{abstract}
Consider an $n$-vertex graph $G = (V,E)$ of maximum degree $\Delta$, and suppose that each vertex $v
\in V$ hosts a processor. The processors are allowed to communicate only with their neighbors in
$G$. The communication is synchronous, i.e., it proceeds in discrete rounds.

In the distributed {\em vertex coloring} problem the objective is to color $G$ with $\Delta + 1$, or
slightly more than $\Delta + 1$, colors using as few rounds of communication as possible.
(The number of rounds of communication will be henceforth referred to as {\em running time}.)
Efficient {\em randomized} algorithms for this problem are known for more than twenty years
\cite{L86, ABI86}. Specifically, these algorithms produce a $(\Delta + 1)$-coloring within $O(\log
n)$ time, with high probability. On the other hand, the best known {\em deterministic} algorithm that requires polylogarithmic time employs $O(\Delta^2)$ colors. This algorithm was devised in a seminal FOCS'87 paper by Linial \cite{L87}. Its running time is $O(\log^* n)$. In the same paper Linial asked whether one can color with significantly less than $\Delta^2$ colors in deterministic polylogarithmic time. By now this question of Linial became one of the most central long-standing open questions in this area.

In this paper we answer this question in the affirmative, and devise a deterministic algorithm that employs
$\Delta^{1 +o(1)}$ colors, and runs in {\em polylogarithmic} time. Specifically, the running time of our algorithm is $O(f(\Delta) \log
\Delta \log n)$, for an arbitrarily slow-growing function $f(\Delta) = \omega(1)$.
We can also produce $O(\Delta^{1 + \eta})$-coloring in $O(\log \Delta \log n)$-time, for an arbitrarily small constant $\eta > 0$, and $O(\Delta)$-coloring in $O(\Delta^{\epsilon} \log n)$ time, for an arbitrarily small constant $\epsilon > 0$. Our results are, in fact, far more general than
this. In particular, for a graph of arboricity $a$, our algorithm produces an $O(a^{1 + \eta})$-coloring, for an arbitrarily small constant $\eta > 0$, in time $O(\log a \log n)$. \\ \\ \\ \\ \\ \\ \\ \\ \\

\end{abstract}

\end{titlepage}

\section{Introduction}
{\bf 1.1 Background and Previous Research}

\noindent In the message passing model of distributed computing the network is modeled by an $n$-vertex undirected unweighted graph $G = (V,E)$, with each vertex hosting its own processor with a unique identity number. These numbers are assumed to belong to the range $\{1,2,...,n\}$. Initially, each vertex $v$ knows only its identity number $id(v)$. The vertices communicate over the edges of $E$ in the {\em synchronous manner}. Specifically, computations (or equivalently, algorithms) proceed in discrete rounds. In each round each vertex $v$ is allowed to send a message to each of its neighbors. All messages that are sent in a certain round arrive to their destinations before the next round starts. The number of rounds that elapse from the beginning of the algorithm until its end is called the {\em running time} of the algorithm.

In the {\em vertex coloring} problem one wants to color the vertices of $V$ in such a way that no edge becomes monochromatic. It is very easy to color a graph $G$ of maximum degree $\Delta = \Delta(G)$ in $\Delta + 1$ colors using $n$ rounds. Coloring it in $\Delta + 1$, or slightly more than $\Delta + 1$, colors far more efficiently is one of the most central and fundamental problems in distributed computing. In addition to its theoretical appeal, the problem is also very well motivated by various real-life network tasks \cite{HT04,P00}. 

The vertex coloring problem is also closely related to the {\em maximal independent set} (henceforth, MIS) problem. A subset $U \subseteq V$ is an {\em independent set} if there is no edge $(u,u') \in E$ with both endpoints in $u$. It is an MIS if for every vertex $v \in V \setminus U$, the set $U \bigcup \{v\}$ is not an independent set. A classical reduction of Linial \cite {L92} shows that given a (distributed) algorithm for computing an MIS on general graphs, one can obtain a $(\Delta + 1)$-coloring within the same time.

The (distributed) vertex coloring and MIS problems have been intensively studied since the mid-eighties. Already in 1986 Luby \cite{L86} and Alon, Babai and Itai \cite{ABI86} devised randomized algorithms for the MIS problem that require $O(\log n)$ time. Using Linial's reduction \cite{L92} these results imply that $(\Delta + 1)$-coloring can also be computed in randomized logarithmic time. More recently, Kothapalli et al. \cite{KSOS06} devised a randomized $O(\Delta)$-coloring algorithm that requires $O(\sqrt{\log n })$ time. 
On the other hand, the best known {\em deterministic} algorithm that requires polylogarithmic time employs $O(\Delta^2)$ colors. Specifically, its running time is $O(\log^* n)$\footnote[1]{$\log^* n$ is the smallest integer $t$ such that the $t$-iterated  logarithm of $n$ is no greater than 2, i.e., $\log^{(t)} n \leq 2$.}. This algorithm was devised in a seminal FOCS'87 paper of Linial \cite{L92}. In the end of this paper Linial argued that his method cannot be used to reduce the number of colors below  ${\Delta + 2}\choose{2}$, and asked whether this can be achieved by other means. Specifically, he wrote

\textit{"Proposition 3.4 of [EFF] shows that set systems of the type that would allow further reduction of the number of colors do not exist. Other algorithms may still be capable of coloring with fewer colors.  \textbf{It would be interesting to decide whether this quadratic bound can be improved when time bounds rise from $O(\log^* n)$ to polylog, for instance.}"}

By now, almost quarter a century later, this open question of Linial became one of the most central long-standing open questions in this area.

\noindent {\bf 1.2 Our Results}

\noindent In this paper we answer this question in the affirmative. Specifically, for an arbitrarily small constant $\eta > 0$, our algorithm constructs an $O(\Delta^{1 + \eta})$-coloring in $O(\log \Delta \log n)$ time. Moreover, we show that one can trade time for the number of colors, and devise a $\Delta^{1 + o(1)}$-coloring algorithm with running time $O(f(\Delta) \log \Delta \log n)$, where $f(\Delta) = \omega(1)$ is an arbitrarily  slowly-growing function of $\Delta$. Also, our algorithm can produce an $O(\Delta)$-coloring in $O(\Delta^{\epsilon} \log n)$ time, for an arbitrarily small constant $\epsilon > 0$.

Currently, the state-of-the-art bound for deterministic $O(\Delta^{1 + \eta})$-coloring (respectively, $\Delta^{1 + o(1)}$-coloring; resp., $O(\Delta)$-coloring) is $\min\{O(\Delta^{1-\eta} + \log^* n), 2^{O(\sqrt{\log n})}\}$ (respectively, $\min\{\Delta^{1-o(1)} + O(\log^* n), 2^{O(\sqrt{\log n})}\}$; resp., $\min\{O(\Delta + \log^* n), 2^{O(\sqrt{\log n})}\}$). (Algorithms that produce $O(\Delta \cdot t)$-coloring in $O(\Delta/t + \log^* n)$ time were devised in \cite{BE09} and in \cite{K09}. The algorithm of \cite{PS95} requires $2^{O(\sqrt{\log n})}$ time.) Our results constitute an {\em exponential} improvement over this state-of-the-art for large values of $\Delta$ (i.e., $\Delta = 2^{\Omega(\log^{\epsilon} n)}$, for some constant $\epsilon > 0$), and a significant improvement in the wide range of $\Delta = \log^{1 + \Omega(1)} n$.

In addition, our results are, in fact, far more general than described above. Specifically, we consider graphs with bounded {\em arboricity} \footnote[1]{ The {\em arboricity} of a graph $G = (V,E)$ is the minimal number $a$ such that the edge set $E$ of $G$ can be covered with at most $a$ edge disjoint forests.}
rather than bounded degree. This is a much wider family of graphs that contains, in addition to graphs of bounded degree, the graphs of bounded genus, bounded tree-width, graphs that exclude a fixed minor, and many other graphs. All the results that we have stated above apply to graphs of arboricity at most $a$. (One just needs to replace $\Delta$ by $a$ in the statements of all results. We remark that a graph with maximum degree $\Delta$ has arboricity at most $\Delta$ as well.) 
One interesting consequence of this extension is that if the arboricity $a$ and 
the degree $\Delta$
of a graph  are polynomially separated one from another (i.e., if there 
exists a constant $\nu > 0$
such that $a \le \Delta^{1 - \nu}$) then our algorithm constructs a $(\Delta 
+ 1)$-coloring
\footnote[2]{Actually, even $o(\Delta)$-coloring.} in $O(\log a \cdot \log n) = 
O(\log \Delta \cdot \log n)$
time.

We also show that one can decrease the running time further almost all the way to $\log n$, while still having less than $a^2$ colors. Specifically, we show that in $O(\log \log a \log n)$ time one can construct an $O(a^2 / \log^C a)$-coloring, for an arbitrarily large constant $C$. More generally, for any function $\omega(1) = f(a) = o(\log a) $, one can construct an $O(a^2 / 2^{f(a)})$-coloring in $O(f(a) \log n)$ time.

Our algorithms for coloring graphs of arboricity $a$ also compare very favorably with the current state-of-the-art. Specifically, the fastest algorithm known today  for $O(a)$-coloring \cite{BE08} requires $O(a \log n)$ time. Our algorithm produces $O(a)$-coloring in $O(a^{\epsilon} \log n)$ time, for arbitrarily small constant $\epsilon > 0$. The best known tradeoff between the number of colors and the running time (also due to \cite{BE08}) is $O(a \cdot t)$-coloring in $O(\frac{a}{t} \log n + \log n + a)$ time. We improve this tradeoff and show that $O(a \cdot t)$-coloring can be computed in just $O((\frac{a}{t})^{\epsilon} \log n)$ time, for an arbitrarily small constant $\epsilon > 0$. In some points on the tradeoff curve the improvement is even greater than that. For example, we compute an $O(a^{1 + \eta})$-coloring, for an arbitrarily small $\eta > 0$, in $O(\log a \log n)$ time, while the previous bound was $O(a^{1 - \eta} \log n + a)$ time. Similarly, our $a^{1 + o(1)}$-coloring algorithm requires $O(f(a) \log a \log n)$ time, for any function $f(a) = \omega(1)$, while the previous bound was $a^{1 - o(1)} \cdot \log n + O(a)$ time.

Finally, our results imply improved bounds for the deterministic MIS problem on graphs of bounded arboricity $a$. Specifically, our algorithm produces an MIS in time $O(a + a^{\epsilon} \log n)$, for an arbitrarily small constant $\epsilon > 0$. The previous state-of-the-art is $\min \{O(a \sqrt{\log n} + \log n), \ 2^{O(\sqrt{\log n})}\}$ due to \cite{BE08, PS95}. Hence our result is stronger in the wide range $\log^{1/2 + \Omega(1)} n \leq a \leq 2^{c \sqrt {\log n}}$, for some universal constant $c > 0$.


\noindent {\bf 1.3 Our Techniques and Overview of the Proof}

\noindent We employ a combination of a number of different existing approaches, together with a number of novel ideas. The first main building block is the machinery for distributed forests decomposition, developed by us in a previous paper \cite{BE08}. Specifically, it is known \cite{BE08} that a graph $G = (V,E)$ of arboricity $a$ can be efficiently distributedly decomposed into $O(a)$ edge-disjoint forests in $O(\log n )$ time. Moreover, these forests come with a {\em complete acyclic orientation} of the edges of $E$. In other words, both endpoints $u$ and $v$ of every edge $e = (u,v) \in E$ know the identity of the forest $F$ to which the edge $e$ ends up to belong, and the parent-child relation of $u$ and $w$ in $F$. In addition, this forests decomposition comes along with another useful graph decomposition, called {\em $H$-partition}. Roughly speaking, an $H$-partition is a decomposition of the {\em vertex} set $V$ of $G$ into $\ell = O(\log n)$ vertex sets $H_1,H_2,...,H_{\ell}$, such that each $G(H_i)$, $i = 1,2,...,\ell$ is a graph with maximum degree $O(a)$. (See Section 2.2 for more details.) This decomposition is extremely useful, as it allows one to apply algorithms that were devised for graphs of bounded degree on graphs of bounded arboricity. We will discuss this point further below.

The second main building block is the suite of algorithms for constructing {\em defective colorings}, developed by us in another previous paper (in STOC'09 \cite{BE09, BE09ar}), and by Kuhn (in SPAA'09 \cite{K09}). These algorithms enable one to efficiently decompose a graph $G$ of maximum degree $\Delta$ into $h = O(t^2)$ subgraphs $G'_1,G'_2,...,G'_h$, of maximum degree $\Delta' = O(\Delta/t)$ each. This decomposition was used in \cite{BE09,K09} for devising $(\Delta + 1)$-coloring algorithms that run in $O(\Delta + \log^* n)$ time. It is a natural idea to construct this decomposition and then to recurse on each of the subgraphs. However, unfortunately, the product $h \cdot \Delta'$ may be significantly larger than $\Delta$. Consequently, in this simplistic form this approach is doomed either to have a large running time or to use prohibitively many colors.

The approach that we employ in this paper is based on {\em arbdefective colorings}. While defective coloring is a classical graph-theoretic notion \cite{AJ85, CCW86, HJ85}, arbdefective coloring is a new concept that we introduce in this paper. It generalizes the notion of defective coloring. A coloring $\varphi$ is an {\em $r$-arbdefective $k$-coloring} if it employs $k$ colors, and each color class induces a subgraph of {\em arboricity} at most $r$.
We demonstrate that in a graph $G$ of arboricity $a$, an $r$-arbdefective $k$-coloring with $r \cdot k = O(a)$ can be efficiently computed. Here the combination of parameters is significantly better than in the case of defective colorings, and consequently, recursing on each of the subgraphs gives rise to an efficient algorithm for $O(a)$-coloring of the original graph $G$.

Therefore, the heart of our proof is an efficient algorithm for computing arbdefective colorings. Our algorithm for this task works in the following way. First, it computes an $H$-partition of the input graph $G$, i.e., it decomposes the vertex set of $G$ into subgraphs $H_1,H_2,...,H_{\ell}$ such that every subgraph has maximum degree of $O(a)$. Second, it computes $O(a/t)$-defective $t^2$-coloring $\varphi_i$ for each of the subgraphs $H_i$. Finally, it employs these colorings to compute a unified $O(a/t)$-arbdefective $t$-coloring $\varphi$ for the entire graph $G$. In other words, we show that a graph of small maximum {\em degree} can be efficiently colored by a good {\em arbdefective} coloring. We use this fact to construct good arbdefective colorings for graphs with bounded arboricity, by first decomposing these graphs into parts that have small maximum degrees.

We believe that this interplay between graphs with bounded maximum degree and graphs with bounded arboricity is very interesting. Our algorithm essentially "zig-zags" between the two families of graphs. First, it decomposes a graph with bounded arboricity into many subgraphs with bounded degree. Then it decomposes each of these subgraphs into subgraphs with bounded arboricity. It then merges these subgraphs in a certain subtle way, to obtain another decomposition into subgraphs with bounded arboricity. The algorithm then recurses on each of these subgraphs.

Another intricate part of our argument is the routine that computes an $O(\Delta/t)$-arbdefective $t$-coloring of a graph $G$ with maximum degree at most $\Delta$. In this part of the proof we manipulate with orientations in a novel way. A {\em complete orientation} $\sigma$ assigns a direction to each edge $e = (u,w)$ of $G$. Orientations play a central role in the theory of distributed graph coloring  \cite{CV86, L92, KSOS06}.
We introduce the notion of {\em partial orientations}. A partial orientation $\sigma$ is allowed not to orient some edges of the graph. Specifically, we say that $\sigma$ has {\em deficit} at most $d$, for some positive integer parameter $d$, if for every vertex $v$ in the graph the number of edges incident to $v$ that $\sigma$ does not orient is no greater than $d$. Another important parameter of an orientation $\sigma$ is its {\em length}, defined as the length of the longest path $P$ in which all edges are oriented consistently according to $\sigma$.
We demonstrate that partial orientations with appropriate deficit and length parameters can be constructed efficiently. Moreover, these orientations turn out to be extremely useful for computing arbdefective colorings. We believe that the notion of partial orientation, and our technique of constructing these orientations are of independent interest.

\noindent {\bf 1.4 Related Work}

\noindent There is an enormous amount of literature on distributed graph coloring. Already before the work of Linial, Cole and Vishkin \cite{CV86} devised a deterministic $3$-coloring algorithm with running time $O(\log^* n)$ for oriented \footnote[1]{In an {\em oriented} ring each vertex $v$ knows which of its two neighbors is located in the clockwise direction from $v$, and which is located in the counter-clockwise direction.} rings. In STOC'87 Goldberg and Plotkin \cite{GP87} generalized the algorithm of \cite{CV86} and obtained a $(\Delta + 1)$-coloring algorithm with running time $2^{O(\Delta)} + O(\log^* n)$. Also, Goldberg, Plotkin, and Shannon \cite{GPS88} devised a $(\Delta + 1)$-coloring algorithm with running time $O(\Delta \log n)$. In FOCS'89 Awerbuch, Goldberg, Luby and Plotkin \cite{AGLP89} devised a $2^{O(\sqrt{\log n \log \log n})}$-time deterministic algorithm for the MIS, and consequently, for the $(\Delta + 1)$-coloring problem. In STOC'92 Panconesi and Srinivasan \cite{PS95} improved this upper bound to $2^{O(\sqrt{\log n})}$. More recently, in PODC'06 Kuhn and Wattenhofer \cite{KW06} devised a $(\Delta + 1)$-coloring algorithm with running time $O(\Delta \log \Delta + \log^* n)$. In STOC'09 Barenboim and Elkin \cite{BE09, BE09ar}, and independently Kuhn \cite{K09} in SPAA'09, devised a $(\Delta + 1)$-coloring algorithm with running time $O(\Delta + \log^* n)$.

Another related thread of study is the theory of distributed graph decompositions. (See the book of Peleg \cite{P00} for an excellent in-depth survey of this topic.) In particular, Awerbuch et al. \cite{AGLP89} and Panconesi and Srinivasan \cite{PS95} showed that any $n$-vertex graph $G$ can be efficiently decomposed into disjoint regions of diameter $2^{O(\sqrt{\log n})}$, so that the super-graph induced by contracting each region into a super-vertex has arboricity $2^{O(\sqrt{\log n})}$. Linial and Saks \cite{LS92} proved another important related result of this kind. Specifically, they showed that $G$ can be decomposed into regions of diameter $O(\log n)$, so that the induced super-graph has chromatic number $O(\log n)$. Both these results were used in \cite{AGLP89, LS92, PS95} for devising efficient coloring and MIS algorithms. (The algorithms of \cite{AGLP89, PS95} are deterministic, and the algorithm of \cite{LS92} is randomized.)
Note, however, that these decompositions are inherently different from the ones that we develop, in a number of ways.


We remark that the algorithmic scheme of \cite{AGLP89, LS92, PS95} that utilizes graph decompositions stipulates that on each round only a small portion of all vertices (specifically, vertices that belong to regions of a given color) are active. This approach is inherently suboptimal, as it does not exploit the network parallelism to the fullest possible extent. In our approach, on the contrary, once the original graph is decomposed into subgraphs the algorithm recurses {\em in parallel on all subgraphs}. In this way all vertices are active at (almost) \footnote[1]{In some branches the recursion may proceed faster than in others. This may result in some vertices becoming idle sooner than other vertices.} all times. This extensive utilization of parallelism is the key to the drastically improved running time of our algorithms.
 

\section{Preliminaries}

\noindent {\bf 2.1 Definitions and Notation}

\noindent Unless the base value is specified, all logarithms in this paper are to base 2.\\
The {\em out-degree} of a vertex $v$ in a directed graph is the number of edges incident to $v$ that are oriented out of $v$.
An {\em orientation} $\sigma$ of (the edge set of) a  graph is an assignment of direction to each edge $(u,v) \in E$, either towards $u$ or towards $v$. 
A {\em partial orientation} is an orientation of a subset $E' \subseteq E$. Edges in $E \setminus E'$ have no orientation. 
The {\em length} of a vertex $v$ with respect to an orientation $\sigma$, denoted $len(v)$ $=len_{\sigma}(v)$, is the length $\ell$ of the longest directed path $<v = v_0,v_1,...,v_{\ell}>$ that emanates from $v$, where all edges $(v_i,v_{i+1})$, for $i = 0,1,...\ell-1$, are oriented by $\sigma$ towards $v_{i+1}$.
The {\em length} of a (partial) orientation $\sigma$, denoted $len(\sigma)$, is the maximum length of a vertex $v$ with respect to the orientation. 
The {\em deficit} of a vertex $v$ with respect to a partial orientation $\sigma$ is the number of edges $e$ that are unoriented by $\sigma$, and incident to $v$. The {\em deficit of $\sigma$} is the maximum deficit of a vertex $v \in V$ with respect to $\sigma$.
The {\em out-degree} of an orientation $\sigma$ of a graph $G$ is the maximum out-degree of a vertex in $G$ with respect to $\sigma$. 
In a given orientation, each neighbor $u$ of $v$ that is connected to $v$ by an edge oriented towards $u$  is called a {\em parent} of $v$.
In this case we say that $v$ is a {\em child} of $u$.\\
A coloring $\varphi: V \rightarrow \MathN$ that satisfies $\varphi(v) \neq \varphi(u)$ for each edge $(u,v) \in E$ is called a {\em legal coloring}. The minimum number of colors that can be used in a legal coloring of a graph $G$ is called {\em the chromatic number} of $G$. It is denoted  $\chi(G)$.\\
An {\em $m$-defective $p$-coloring} of a graph $G$ is a coloring of the vertices of $G$ using $p$ colors, such that each vertex has at most $m$ neighbors colored by its color. Each color class in the $m$-defective coloring induces a graph of maximum degree $m$. It is known that for any positive integer parameter $p$, an $\left \lfloor \Delta / p \right \rfloor$-defective $O(p^2)$-coloring can be efficiently computed distributively \cite{BE09, BE09ar,K09}.
\begin{lem} \label{dcol} \cite{K09}
A $\left \lfloor \Delta / p \right \rfloor$-defective $O(p^2)$-coloring can be computed in $O(\log^* n)$ time. 
\end{lem}

We conclude this section by defining the notion of arbdefective coloring.
This notion generalizes the notion of defective coloring.

\begin{dfn}
An {\em $r$-arbdefective $k$-coloring} is a coloring with $k$ colors such that all the vertices colored by the same color $i$, $1 \leq i \leq k$, induce a subgraph of arboricity at most $r$.
\end{dfn}

\noindent {\bf 2.2 Forests-Decomposition}

\noindent A {\em $k$-forests-decomposition} is a partition of the edge set of the graph into $k$ subsets, such that each subset forms a forest.
Efficient distributed algorithms for computing $O(a)$-forests decompositions have been devised recently in \cite{BE08}. Several results from \cite{BE08} are used in the current paper. They are summarized in the following lemmas.
\begin{lem} \label{arbcol} \cite{BE08}
(1) For any graph $G$, a legal $(\left\lfloor (2 + \epsilon ) \cdot a \right \rfloor + 1)$-coloring of $G$ can be computed in $O(a \log n)$ time, for an arbitrarily small positive constant $\epsilon$.

(2)
For any graph $G$, an $O(a)$-forests-decomposition can be computed in $O(\log n)$ time.
\end{lem}
Moreover, the algorithm in \cite{BE08} for computing forests-decompositions produces a vertex partition with a certain helpful property, called an $H$-partition. An {\em $H$-partition} is a partition of $V$ into subsets $H_1,H_2,...,H_{\ell}$, $\ell = O(\log n)$, such that each vertex in $H_i$, $1 \leq i \leq \ell$, has at most $O(a)$ neighbors in $\bigcup_{j = i}^{\ell} H_j$. The {\em degree of the $H$-partition} is the maximum number of neighbors of a vertex $v \in H_i$ in $\bigcup_{j = i}^{\ell} H_j$ for $1 \leq i \leq \ell$. For a vertex $v \in V$, the {\em $H$-index} of $v$ is the index $i$, $ 1 \leq i \leq \ell$, such that $v \in H_i$.
\begin{lem} \label{hpart} \cite{BE08}
For any graph $G$, an $H$-partition of the vertex set of $G$ can be computed in $O(\log n)$ time. The degree of the computed $H$-partition is $\left \lfloor (2 + \epsilon) \cdot a \right \rfloor$, for an arbitrarily small positive constant $\epsilon$.
\end{lem}
The $H$-partition is used to compute an acyclic orientation such that each vertex has out-degree $O(a)$.
\begin{lem} \cite{BE08}
For any graph $G$, an acyclic complete orientation with out-degree $O(a)$ can be computed in $O(\log n)$ time.
\end{lem}
Finally, the relationship between arboricity and acyclic orientation is given in the following lemma. 
\begin{lem} \label{acyclic} \cite{BE08, E94}
If there exists an acyclic complete orientation of $G$ with out-degree $k$, then $a(G) \leq k$.
\end{lem}


\section{Small Arboricity Decomposition}

We begin with presenting a simple algorithm that computes an $O(a/k)$-arbdefective $k$-coloring for any integer parameter $k > 0$. 
(In other words, it computes a vertex decomposition into $k$ subgraphs such that each subgraph has arboricity $O(a/k)$.) 
The running time of our first algorithm is $O(a \log n)$. Later, we present an improved version of the algorithm with a significantly faster running time.

Suppose that we are given an acyclic complete orientation of the edge set of $G$, such that each vertex has at most $m$ outgoing edges, for a positive parameter $m$. The following procedure, called {\em Procedure Simple-Arbdefective} accepts as input such an orientation and a positive integer parameter $k$. During its execution, each vertex computes its color in the following way. The vertex waits for all its parents to select their colors. (Recall that a parent of a vertex $v$ is a neighbor $u$ of $v$ connected by an edge $\langle v,u \rangle$ that is oriented towards $u$.) Once the vertex receives a message from each of its parents containing their selections, it selects a color from the range $\{ 1,2,...,k \}$ that is used by the minimum number of parents. Then it sends its selection to all its neighbors. This completes the description of the procedure.

Let $c$ be the color that a vertex $v$ has selected. Since $v$ has at most $m$ parents, by the Pigeonhole Principle, the number of parents colored by the color $c$ is at most $\left \lfloor m/k \right \rfloor$.  For $c = 1,2,...,k$, consider the subgraph $G_c$ induced by all the vertices that have selected the color $c$. For each edge $e$ in $G_c$, orient $e$ in the same way it is oriented in the original graph $G$. The orientation in $G_c$ is therefore acyclic, and each vertex in $G_c$ has out-degree at most $\left \lfloor m/k \right \rfloor$. Thus, the arboricity of $G_c$ is at most $\left \lfloor m/k \right \rfloor$ (See Lemma \ref{acyclic}). Hence Procedure Simple-Arbdefective has produced an $\left \lfloor m/k \right \rfloor$-arbdefective $k$-coloring.


Next, we consider a more general scenario in which instead of accepting as input a complete orientation we are given a partial orientation. Specifically, the orientation that Procedure Simple-Arbdefective accepts as input has out-degree at most $m$, and deficit at most $\tau$.
Once the procedure is invoked on such an orientation and a parameter $k$ as input, a coloring with $k$ colors is produced. Consider the graph $G_c$ induced by all the vertices that are colored by the color $c$, $1 \leq c \leq k$. Each edge in $G_c$ is oriented in the same way as in $G$. Each vertex in $G_c$ has at most $\left \lfloor m/k \right \rfloor$ parents, and at most $\tau$ unoriented edges connected to it in $G_c$. The following lemma states that it is possible to orient all unoriented edges of $G_c$ to achieve a complete acyclic orientation.

\begin{lem} \label{compacyclic}
Any acyclic partial orientation $\sigma$ of a graph $G = (V,E)$ can be transformed into a complete acyclic orientation by adding orientation to unoriented edges.
\end{lem}
\begin{proof}
Let $\hat{E}$ be the set of all edges oriented by $\sigma$. Since $\sigma$ is acyclic, the graph $\hat{G} = (V,\hat{E})$ is a directed acyclic graph. Perform a topological sort of $\hat{G}$ such that for any edge $\langle u,v \rangle$ that is oriented towards $v$, the vertex $v$ is placed after $u$. Orient each unoriented edge $(w,z)$ in $G$ towards the endpoint that appears later in the topological sorting of $\hat{G}$. It is easy to see that the resulting orientation is a complete acyclic orientation of $G$.
\end{proof}

Once the unoriented edges of $G_c$ are oriented as in the proof of Lemma \ref{compacyclic}, each vertex $v$ in $G_c$ has an out-degree at most $\tau + \left\lfloor m/k \right \rfloor$. (Recall that $v$ had at most $\tau$ unoriented edges incident to it.) Hence, by Lemma \ref{acyclic}, the arboricity of $G_c$ is at most $\tau + \left\lfloor m/k \right \rfloor$.
The next Theorem summarizes the properties of Procedure Simple-Arbdefective.


\begin{thm} \label{simp-arb}
Suppose that Procedure Simple-Arbdefective is invoked with the following two input arguments:
(1) An acyclic (partial) orientation of length $\ell$, out-degree at most $m$, and deficit at most $\tau$.  \\
(2) An integer parameter $k > 0$. \\
Then Procedure Simple-Arbdefective produces a \  $(\tau + \left\lfloor m/k \right \rfloor)$-arbdefective $k$-coloring in $O(\ell)$ time. 
\end{thm}

\begin{proof}
By Lemmas \ref{acyclic}, \ref{compacyclic}, the arboricity of $G_c$, for $1 \leq k \leq c$, is at most $\tau + \left\lfloor m/k \right \rfloor$. Hence, Procedure Simple-Arbdefective produces a $(\tau + \left\lfloor m/k \right \rfloor)$-arbdefective $k$-coloring.

Next, we analyze the running time of Procedure Simple-Arbdefective. We prove by induction on $i$ that after $i$ rounds, all vertices $v \in V$ with $len(v) \leq i$, have selected their colors.

\textbf{Base $(i = 0)$:} The vertices $v$ with $len(v) = 0$ are the vertices that have no outgoing edges. Such vertices select an arbitrary color from the range $\{1,2,...,k\}$ immediately after the algorithm starts, requiring no communication whatsoever.

\textbf{Induction step:} Assume that after $i-1$ rounds, all vertices $v \in V$ with $len(v) \leq i-1$, have selected their colors. Let $u$ be a vertex with $len(u) \leq i$. Then, for each parent $w$ of $u$, it holds that $len(u) \leq i-1$. Consequently, by the induction hypothesis, all parents of $u$ select their color after at most $i-1$ rounds. Therefore, the vertex $u$ is aware of the selection of all its parents on round $i$ or before. Hence, after $i$ rounds the vertex $u$ necessarily selects a color. This completes the inductive proof.

If Procedure Simple-Arbdefective accepts as input an acyclic orientation of length $\ell$, then all directed paths are of length at most $\ell$. Consequently, all vertices select their color after at most $\ell$ rounds.
\end{proof}


For Procedure Simple-Arbdefective to be useful, we need to compute partial acyclic orientations with small length and out-degree. Next,  we devise efficient algorithms for computing appropriate acyclic orientations. First, we devise a distributed algorithm that receives as input an undirected graph $G$, and computes a complete acyclic orientation such that each vertex has out-degree $O(a)$. Observe that in a distributed computation of an orientation, each vertex has to compute only the orientation of edges incident to it, as long as the global solution formed by this computation is correct. The algorithm we devise is called {\em Procedure Complete-Orientation}.

Procedure Complete-Orientation consists of three steps. First, an $H$-partition of the input graph $G$ is computed. (See Section 2.2.) As a consequence, the vertex set of $G$ is partitioned into $\ell' = O(\log n)$ subsets $H_1,H_2,...,H_{\ell'}$, such that each vertex in $H_i$, $1 \leq i \leq \ell'$, has $O(a)$ neighbors in $\bigcup_{j = i}^{\ell'} H_j$. Next, each subgraph induced by a set $H_i$ is colored legally using $O(a)$ colors. Finally, an orientation is computed as follows. Consider an edge $(u,v)$ such that $u \in H_i$ and $v \in H_j$ for some $1 \leq i,j \leq \ell'$. If $i < j$, orient the edge towards $v$. If $j < i$ orient the edge towards $u$. Otherwise $i = j$. In this case the vertices $u$ and $v$ have different colors. Orient the edge towards the vertex that is colored with a greater color. This completes the description of the procedure. We summarize the properties of Procedure Complete-Orientation in the following lemma.

\begin{lem} \label{comp-orient}
The running time of Procedure Complete-Orientation is $O(a + \log n)$. It produces a complete acyclic orientation with out-degree $\left \lfloor (2 + \epsilon) \cdot a \right \rfloor$ for an arbitrarily small constant $\epsilon > 0$, and length $O(a \log n)$.
\end{lem}
\begin{proof}
By Theorem \ref{hpart}, the first step, in which the $H$-partition is computed, requires $O(\log n)$ time. The second step consists of coloring graphs of maximum degree $O(a)$. All the colorings are performed in parallel in $O(a + \log^* n)$ time using the algorithm from \cite{BE09}. The orientation step requires a single round in which vertices learn the colors and the $H$-indices of their neighbors. To summarize, the total running time is $O(a + \log n)$.

Next, we show that the out-degree of each vertex is $O(a)$. Consider a vertex $v \in H_i$. Each outgoing edge of $v$ is connected to a vertex in a set $H_j$ such that $j \geq i$. By Lemma \ref{hpart}, $v$ has at most $\left \lfloor (2 + \epsilon) \cdot a \right \rfloor$ neighbors in $\bigcup_{j=i}^{\ell'} H_j$. Therefore, the out-degree of $v$ is $\left \lfloor (2 + \epsilon) \cdot a \right \rfloor$.
Next, we show that the length of the orientation is $O(a \log n)$. Consider a subgraph $G_i$ induced by a set $H_i$, $1 \leq i \leq \ell'$. Each edge in $G_i$ is oriented towards a vertex with a greater color among its two endpoints. Hence, a certain color appears in any directed path at most once. Consequently, the length of the longest directed path in $G_i$ is less than the number of colors used for coloring $G_i$. Since $G_i$ is colored using $O(a)$ colors, the length of the longest directed path in $G_i$ is $O(a)$. Consider an edge $(u,v)$ such that $u \in H_i, v \in H_j, i < j$. The edge $(u,v)$ is oriented towards the set with the greater index $H_j$. Therefore, each directed path in $G$ has at most $\ell' - 1$ edges whose endpoints belong to different $H$-sets. Inside any path, each two edges whose endpoints belong to different $H$-sets are separated be $O(a)$ consequent edges whose endpoints belong to the same $H$-set. Therefore, the length of any directed path is $O(a \cdot \ell') = O(a \log n)$.
\end{proof} 

Theorem \ref{simp-arb} and Lemma \ref{comp-orient} give rise directly to the following corollary.
\begin{col}
For an integer $k > 0$, an $O(a/k)$-arbdefective $k$-coloring can be computed in $O(a \log n)$ time.
\end{col}

The running time of Procedure Simple-Arbdefective is proportional to the length of the acyclic orientation that is given to the procedure as part of its input. Hence, to improve its running time, we have to compute a much shorter orientation. However, the shortest complete acyclic orientation of a graph $G$ is of length at least $\chi(G) - 1$. (Since a complete acyclic orientation of length $\ell$ allows one to color the graph legally with $\ell + 1$ colors. See Appendix A for more details.)
There exist graphs for which $\chi(G) = \Omega(a)$. Consequently, an acyclic complete orientation of length $o(a)$ does not always exist. We overcome this difficulty by computing a partial acyclic orientation instead. This partial orientation is significantly shorter, and its deficit is sufficiently small. Moreover, we show that a partial orientation can be computed considerably faster than a complete orientation. Also, in the computation of a partial orientation it is no longer required that the $H$-sets are legally colored, which is the case in Procedure Complete-Orientation. Instead it suffices to color the $H$-sets with a defective coloring, and this can be done far more efficiently. (See Lemma \ref{dcol}.)

The pseudocode of the algorithm for computing short acyclic orientations, called {\em Procedure Partial-Orientation} is given below. It receives as input a graph $G$ and a positive integer parameter $t$. It computes an orientation with out-degree $\left \lfloor (2 + \epsilon) \cdot a \right \rfloor$, and deficit at most $\left \lfloor a/t \right \rfloor$. Procedure Partial-Orientation is similar to Procedure Complete-Orientation, except step 2, in which an $\left\lfloor a/t \right \rfloor$-defective $O(t^2)$-coloring is computed instead of a legal $O(a)$-coloring.

\begin{algorithm}[H]
\caption{Procedure Partial-Orientation($G, t$) }

\label{proced:partial}

\begin{algorithmic}[1] 

\STATE $H_1$,$H_2$,...,$H_{\ell'}$ := an $H$-partition of $G$.

\FOR {$i = 1,2...,\ell'$ in parallel}

   \STATE compute an $\left \lfloor a/t \right \rfloor$-defective  $O(t^2)$-coloring $\varphi$ of $G(H_i)$.
   
\ENDFOR

\FOR {each edge $e = (u,v)$ in $E$ in parallel}

   \IF {$u$ and $v$ belong to different $H$-sets}
       
        \STATE orient $e$ towards the set with greater index.
        
   \ELSE 
      
        \IF {$u$ and $v$ have different colors}
            
            \STATE orient $e$ towards the vertex with greater color among $u,v$.
         
         \ENDIF
   \ENDIF
\ENDFOR

\end{algorithmic}
\end{algorithm}

The dominant term in the running time of Procedure Partial-Orientation is the computation of the $H$-partition that requires $O(\log n)$ time. The other steps are significantly faster, since computing defective colorings in lines 2-4 of the procedure requires $O(\log^* n)$ time, and the orientation step (lines 5-13) requires only $O(1)$ time. Therefore, the running time of Procedure Partial-Orientation is $O(\log n)$. Another important property of Procedure Partial-Orientation is that the length of the produced orientation is bounded. Consider a directed path in a subgraph $G(H_i)$, $ 1 \leq i \leq \ell'$. The length of this path is smaller than the number of colors used in the defective coloring of $G(H_i)$, which is $O(t^2)$. Now consider a directed path in the graph $G$ with respect to the orientation produced by Procedure Partial-Orientation. The path contains $O(\log n)$ edges that cross between different $H$-sets. Between any pair of such edges in the path there are $O(t^2)$ consequent edges whose endpoints belong to the same $H$-set. Hence, the length of a directed path in $G$ is $O(t^2 \log n)$. (See Figure 1 below.) The properties of Procedure Partial-Orientation are summarized in the next theorem.
\begin{thm} \label{part-orient}
Let $\epsilon$ be an arbitrarily small positive constant. Procedure Partial-Orientation invoked on a graph $G$ and an integer parameter $t > 0$ produces an acyclic orientation of out-degree $\left \lfloor (2 + \epsilon) \cdot a \right \rfloor$, length $O(t^2 \cdot \log n)$, and deficit at most $\left \lfloor a /t \right \rfloor$. The running time of the procedure is $O(\log n)$.
\end{thm}
\clearpage
\begin{fig}
A directed path $ P = (v_{11},v_{12},...,v_{\ell 4})$ with respect to the orientation produced by Algorithm 1. In this example each $H_i$ is colored with 4 colors. For all $i \in \{1,2,...,\ell\}$, $j \in \{1,2,3,4\}$, $v_{ij}$ is colored by $j$. P contains at most $\ell -1 = O(\log n)$ edges that cross between $H_i$'s.
\end{fig}
\includegraphics{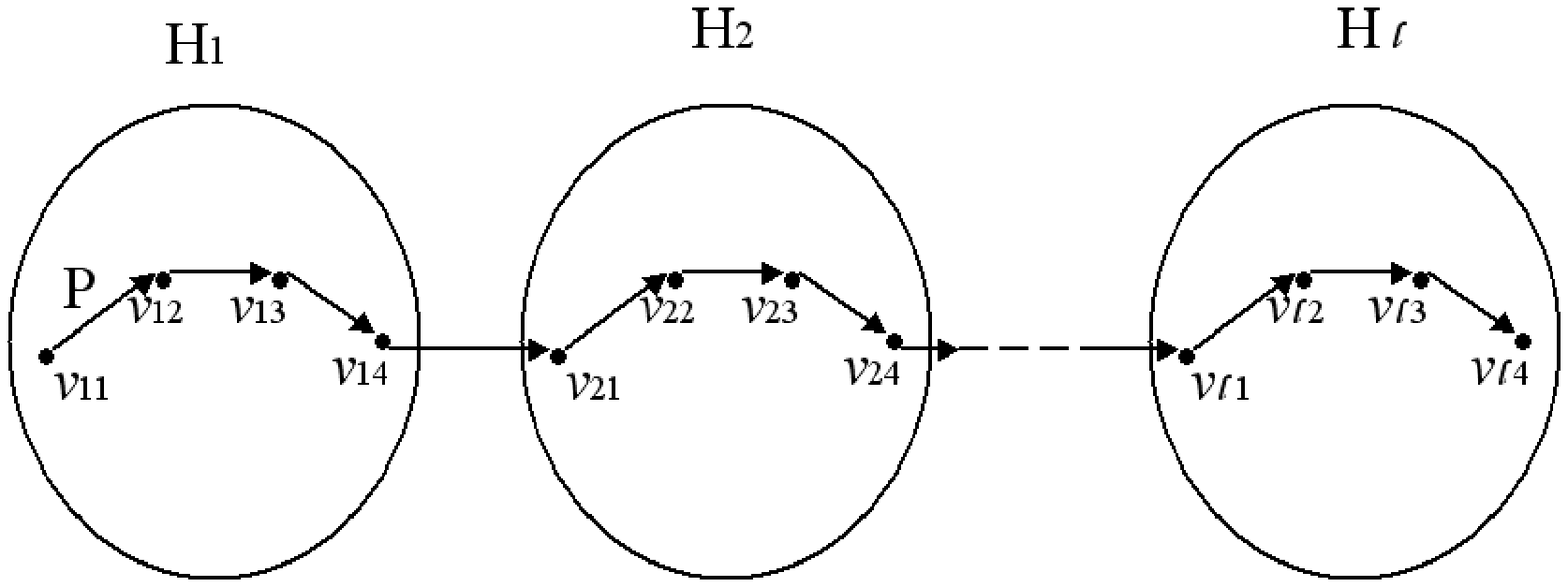}

We conclude this section with an efficient procedure for computing an arbdefective coloring. The procedure, called {\em Procedure Arbdefective-Coloring}, receives as input a graph $G$ and two positive integer parameters $k$ and $t$. First, it invokes Procedure Partial-Orientation on $G$ and $t$. Then it employs the produced orientation and the parameter $k$ as an input for Procedure Simple-Arbdefective, which is invoked once Procedure Partial-Orientation terminates.  This completes the description of the procedure. By Theorem \ref{simp-arb}, the procedure produces an $\left \lfloor a/t +  (2 +  \epsilon)  \cdot a/k \right \rfloor$-arbdefective $k$-coloring. The properties of Procedure Arbdefective-Coloring are summarized in the next corollary. The corollary follows directly from Theorems \ref{simp-arb} and \ref{part-orient}.
\begin{col} \label{arb}
Procedure Arbdefective-Coloring invoked on a graph $G$ and two positive integer parameters $k$ and $t$ computes an $\left \lfloor a/t +  (2 +  \epsilon)  \cdot a/k \right \rfloor$-arbdefective $k$-coloring in time $O(t^2 \log n)$.
\end{col}

We will invoke Procedure Arbdefective-Coloring with $t = k$. In this case it returns a $\left \lfloor (3 + \epsilon)\cdot a / t \right \rfloor$-arbdefective $t$-coloring in $O(t^2 \log n)$ time. Observe that this $t$-coloring can also be viewed as a decomposition of the original graph $G$ into $t$ subgraphs $G'_1,G'_2,...,G'_t$, each of arboricity at most $\left \lfloor (3 + \epsilon)\cdot a / t \right \rfloor$.

\section{Fast Legal Coloring}
In this section we employ the procedures presented in the previous section to devise efficient algorithms that produce legal colorings (i.e., colorings with no defect). Our algorithms rely on the following key properties of arbdefective coloring. Consider a $b$-arbdefective $k$-coloring for some positive integer parameters $b$ and $k$. For $1 \leq i \leq k$, let $G_i$ denote the subgraph induced by all vertices colored with the color $i$. For all $1 \leq i \leq k$, it holds that $a(G_i) \leq b$. Therefore, by Lemma \ref{arbcol}, each subgraph $G_i$ can be efficiently legally colored  using $\left \lfloor (2 + \epsilon) \cdot b \right \rfloor + 1$ colors. If each subgraph is assigned a distinct color palette of size $\left \lfloor (2 + \epsilon) \cdot b \right \rfloor + 1$, then the parallel legal coloring of all subgraphs results in a legal $O(b \cdot k)$-coloring of the entire graph $G$. Observe that once the arbdefective coloring is computed, each vertex $v \in G_i$ communicates only with its neighbors in the same subgraph $G_i$. Once the legal colorings of all subgraphs are computed, the color of $v$ is different not only from all its neighbors in $G_i$, but from all its neighbors in $G$, as we shortly prove.

Our goal is to efficiently compute an $O(a)$-coloring of the graph $G$. Therefore, we employ Corollary \ref{arb} with appropriate parameters to guarantee that $b \cdot k = O(a)$.
First, we present an $O(a)$-coloring algorithm with running time $O(a^{2/3} \log n)$ that involves a single invocation of Procedure Arbdefective-Coloring. Then, we show a more complex algorithm that achieves running time $O(a^{\mu} \log n)$ for an arbitrarily small positive constant $\mu$.

In our first algorithm we invoke Procedure Arbdefective-Coloring on a graph $G$ with the input parameters $k = t = \left \lceil a^{1/3} \right \rceil$. By Corollary \ref{arb}, as a result of this invocation we achieve a   $\left \lfloor  (3 +  \epsilon)  \cdot a^{2/3} \right \rfloor$-arbdefective $\left \lceil a^{1/3} \right \rceil$-coloring of $G$. For $1 \leq i \leq k$, let $G_i$ denote the subgraph induced by all vertices colored with the color $i$. The arboricity of $G_i$ is at most $(3 +  \epsilon)  \cdot a^{2/3}$. For $1 \leq i \leq k$, in parallel, color $G_i$ with $\gamma = \left\lfloor (2 + \epsilon)(3 +  \epsilon)  \cdot a^{2/3}\right \rfloor + 1$ colors. (See Lemma \ref {arbcol}). Let $\psi_i$, $1 \leq i \leq k$, denote the resulting colorings. For each index $i$, $\psi_i$ is a legal coloring of $G_i$. However, for a pair of neighboring vertices $v \in G_i$, $w \in G_j$, $i \neq j$, it may happen that $\psi_i(v) = \psi_j(w)$. Finally, each vertex $v \in G_i$ selects a new color $\varphi$ that is computed by  $\varphi(v) = (i - 1) \cdot \gamma + \psi_i(v)$. Intuitively, the color $\varphi(v)$ can be seen as an ordered pair $<i , \psi_i(v)>$. This completes the description of the algorithm. Its correctness and running time are summarized below.

\begin{lem}  \label{first-time}
$\varphi$ is a legal $O(a)$-coloring of $G$ computed in $O(a^{2/3} \log n)$ time.
\end{lem}
\begin{proof}
First, we prove that $\varphi$ is a legal $O(a)$-coloring. Observe that for each vertex $v$, it holds that $1 \leq \varphi(v) \leq k \cdot \gamma$. Since $k = \left \lceil a^{1/3} \right \rceil$, and $\gamma = \left \lfloor (2 + \epsilon)(3 +  \epsilon)  \cdot a^{2/3}\right \rfloor + 1$, it follows that $\varphi(v) = O(a)$, and consequently $\varphi$ is an $O(a)$-coloring. It is left to show that $\varphi$ is a legal coloring. Consider an edge $(u,v)$ in $G$, such that $u \in G_i$, $v \in G_j$. If $i = j$ then $\psi(u) \neq \psi(v)$ and hence also $\varphi(u) \neq \varphi(v)$. Otherwise $i \neq j$, and again $\varphi(u) \neq \varphi(v)$. 

Next, we prove that $\varphi$ is computed in $O(a^{2/3} \log n)$ time.
By Corollary \ref{arb}, the invocation of Procedure Arbdefective-Coloring requires $O(a^{2/3} \log n)$ time. It produces $k$ subgraphs $G_1,G_2,...,G_k$, each with arboricity at most  $\left \lfloor  (3 +  \epsilon)  \cdot a^{2/3} \right \rfloor$. By Lemma \ref{arbcol}, coloring all subgraphs $G_i$, for $1 \leq i \leq k$ in parallel, requires $O(a^{2/3} \log n)$ time as well. The computation of the final coloring $\varphi$ is performed locally, requiring no additional communication. Therefore, the overall running time is $O(a^{2/3} \log n)$.
\end{proof}

Lemma \ref{first-time} shows that this algorithm is already a significant improvement over the best previously known algorithm for $O(a)$-coloring, whose results are summarized in Lemma \ref{arbcol}. Nevertheless, the running time can be improved further by invoking Procedure Arbdefective-Coloring several times. Since Procedure Arbdefective-Coloring produces subgraphs of smaller arboricity comparing to the input graph, it can be invoked again on the subgraphs, producing a refined decomposition, in which each subgraph has even smaller arboricity. For example, invoking the procedure on a graph $G$ with the parameters $k = t = \left \lceil a^{1/6} \right \rceil$, results in an $O(a^{5/6})$-arbdefective $O(a^{1/6})$-coloring. Invoking the Procedure Arbdefective-Coloring with the same parameters again on all the $O(a^{1/6})$ subgraphs induced by the initial arbdefective coloring results in an $O(a^{4/6})$-arbdefective $O(a^{1/6})$-coloring of each subgraph. If distinct palettes are used for each subgraph, the entire graph is now colored with an $O(a^{2/3})$-arbdefective $O(a^{1/3})$-coloring. The running time of this computation is $O(a^{1/3} \log n)$. This computation is much faster than a single invocation of Procedure Arbdefective-Coloring with the parameters $k = t = \left \lceil a^{1/3} \right \rceil$ that yields the same results. However, to obtain a legal coloring of the original graph $G$, each subgraph still has to be colored legally. For the entire computation to be efficient, the arboricity of all subgraphs has to be as small as possible. Therefore we need to invoke Procedure Arbdefective-Coloring more times to achieve an $o(a^{2/3})$-arbdefective coloring. Indeed, applying Procedure Arbdefective-Coloring on each of the $O(a^{1/3})$ subgraphs produces an $O(\sqrt a)$-arbdefective $O(\sqrt a)$-coloring. This, in turn, directly gives rise to an $O(a)$-coloring within $O(\sqrt a \cdot \log n)$ time.

We employ this idea in the following Procedure called {\em Procedure Legal-Coloring}.

\begin{algorithm}[H]
\caption{Procedure Legal-Coloring($G, p$) }

\label{proced:legal}

\begin{algorithmic}[1] 

\STATE $G_1 := G$

\STATE $\alpha := a(G_1)$

\STATE $\cal {G}$ := $\{ G_1 \}$ \ \ \ \ \ \ \ /* The set of subgraphs */

\WHILE {$\alpha > p$}

    \STATE $\hat{\cal{G}} := \emptyset$ \ \ \ \ \ \ \ /* Temporary variable for storing refinements of $\cal {G}$ */
   
    \FOR {each $G_i \in \cal{G}$ in parallel}
       
        \STATE $G'_1,G'_2,...,G'_p$ := Arbdefective-Coloring($G_i$ , $k:=p$ , $t:=p$) \\        /* $G'_j$ is the subgraph induced by all the vertices that are assigned the color $j$ by the arbdefective coloring */
        
        \FOR {$j := 1,2,...,p$ \ in parallel }
        
           \STATE $z := (i-1) \cdot p + j$ \ \ \ \ \ \ \ \ /* Computing a unique index for each subgraph */
        
            \STATE $\hat{G}_z := G'_j$
            
            \STATE $\hat{\cal{G}} := \hat{\cal{G}}$ $ \cup \ \{ \hat{G}_z \}$
            
        \ENDFOR
     
    \ENDFOR
    
    \STATE $\cal{G}$ := $\hat{\cal{G}}$
     
    \STATE  $\alpha := \left \lfloor \alpha / p +  (2 +  \epsilon)  \cdot \alpha/p \right \rfloor$ \ /* The new upper bound for the arboricity of each of the subgraphs */
   
\ENDWHILE

\STATE $A := \left\lfloor (2 + \epsilon) \alpha \right \rfloor + 1$

\FOR {each $G_i \in \cal{G}$ in parallel}

   \STATE color $G_i$ legally using the palette $\{ (i-1) \cdot A + 1, (i-1) \cdot A + 2,...,i \cdot A \}$         \ \ \ \ \  /* Using Lemma \ref{arbcol} */
   
\ENDFOR

\end{algorithmic}
\end{algorithm}

The procedure receives as input a graph $G$ and a positive integer parameter $p$. It proceeds in phases. In the first phase Procedure Arbdefective-Coloring is invoked on the input graph $G$ with the parameters $k := p$ and $t := p$. Consequently, a decomposition into $p$ subgraphs is produced, in which each subgraph has arboricity $O(a/p)$. In each of the following phases Procedure Arbdefective-Coloring is invoked in parallel on all subgraphs in the decomposition that was created in the previous phase. As a result, a refinement of the decomposition is produced, i.e., each subgraph is partitioned into $p$ subgraphs of smaller arboricity. Consequently, after each phase, the number of subgraphs in $G$ grows by a factor of $p$, but the arboricity of each subgraph decreases by a factor of $\Theta(p)$. Hence, the product of the number of subgraphs and the arboricity of subgraphs remains $O(a)$ after each phase. (As long as the number of phases is constant.) Once the arboricities of all subgraphs become small enough, Lemma \ref{arbcol} is used for a fast parallel coloring of all the subgraphs, resulting in a unified legal $O(a)$-coloring of the input graph.

Let $\mu$ be an arbitrarily small positive constant. We show that invoking Procedure Legal-Coloring on $G$ with the input parameter $p := \left \lfloor a^{ \mu / 2} \right \rfloor$ results in an $O(a)$-coloring in $O(a^{\mu} \log n)$ time. The following lemma constitutes the proof of correctness of the algorithm.

We assume  without loss of generality that the arboricity $a$ is sufficiently large to guarantee that $p \geq 16$. (Otherwise, it holds that $a \leq 17^{2/\mu}$, i.e., the arboricity is bounded by a constant. In this case, by Lemma \ref{arbcol}, one can directly compute an $O(1)$-coloring in $O(\log n)$ time).

Let $\alpha_i$ and ${\cal G }_i$ denote the values of the variables $\alpha$ and $\cal{G}$, respectively, in the end of iteration $i$ of the while-loop of Algorithm \ref{proced:legal} (lines 4-16).

\begin{lem} \label{legal1} {\bf(1)} (Invariant for line 16 of Algorithm \ref{proced:legal})
In the end of iteration $i$ of the while-loop, $i = 1,2,...$, each graph in the collection ${\cal G}_i$ has arboricity at most $\alpha_i$. \\
{\bf (2)} The while-loop runs for a constant number of iterations.\\
{\bf (3)} For i =1,2,..., after $i$ iterations, it holds that $\alpha_i \cdot |{\cal G}_i| \leq (3 + \epsilon)^i \cdot a$.
\end{lem}
\begin{proof}
 \textbf{The proof of (1):}
The proof is by induction on the number of iterations. For the base case, observe that after the first iteration, $\cal{G}$ contains $p$ subgraphs produced by Procedure Arbdefective-Coloring. By Corollary \ref{arb}, the arboricity of each subgraph is at most 
$\left \lfloor a/t +  (2 +  \epsilon)  \cdot a/k \right \rfloor = \left \lfloor a/p +  (2 +  \epsilon)  \cdot a/p \right \rfloor = \alpha_1.$

For the inductive step, consider an iteration $i$. By the induction hypothesis, each subgraph in ${\cal G}_{i-1}$ has arboricity at most $\alpha_{i-1}$. During iteration $i$, Procedure Arbdefective-Coloring is invoked on all subgraphs in ${\cal G}_{i-1}$. Consequently, ${\cal G}_i$ contains new subgraphs, each with arboricity at most $\left \lfloor \alpha_{i-1}/p +  (2 +  \epsilon)  \cdot \alpha_{i-1}/p \right \rfloor$, which is exactly the value $\alpha_i$ of $\alpha$ in the end of iteration $i$. (See line 15.) 

 \textbf{The proof of (2):}
In each iteration the variable $\alpha$ is decreased by a factor of at least $b = p / (3 + \epsilon) $. Hence, the number of iterations is at most $\log_b a$. For any $0 < \epsilon < 1/2$, and a sufficiently large $a$, it holds that
$$ \log_b a = \frac{\log a} {\log (p / (3 + \epsilon))} \leq  \frac{\log a}{\log (\frac{1}{4} a^{\mu/2})} =  \frac{2/\mu \cdot \log a^{\mu/2}}{\log a^{\mu/2} - 2} \leq 4/\mu. $$

 \textbf{The proof of (3):}
The correctness of the lemma follows directly from the fact that in each iteration the number $|{\cal{G}}|$ of subgraphs grows by a factor of $p$, and the arboricity of each subgraph decreases by a factor of at least $ p / (3 + \epsilon)$.
\end{proof}
The next theorem follows from Lemma \ref{legal1}.
\begin{thm} \label {sumlegal}
Invoking Procedure Legal-Coloring on a graph $G$ with arboricity $a$ with the parameter $p = \left\lceil a^{\mu/2} \right \rceil$ for a positive constant $\mu < 1$, computes a legal $O(a)$-coloring of $G$ within $O(a^{\mu} \cdot \log n)$ time.
\end{thm}

\begin{proof}
We first prove that the coloring is legal. Observe that the selection of unique indices in line 9 guarantees that any two distinct subgraphs that were added to the same set $\hat{\cal G}$ are colored using distinct palettes. In addition, in each iteration each vertex belongs to exactly one subgraph in ${\cal G}$. Consequently, once the while-loop terminates, each vertex $v$ belongs to exactly one subgraph in $\cal {G}$. Let $G_i \in {\cal{G}}$ be the subgraph that contains $v$. Let $\alpha'$ denote the value of $\alpha$ on line 17 of Algorithm \ref{proced:legal}. As we have seen, the arboricity of $G_i$ is at most $\alpha'$. Hence, $G_i$ is colored legally using a unique palette containing $A = \left \lfloor (2 + \epsilon) \alpha' + 1 \right \rfloor$ colors. Consequently, the color of $v$ is different from the colors of all its neighbors, not only in $G_i$, but in the entire graph $G$.\

Now we analyze the number of colors in the coloring. By Lemma \ref{legal1}, the number of colors employed is $A \cdot |{\cal G}| = (\left\lfloor (2 + \epsilon) \alpha' \right \rfloor + 1) \cdot |{\cal G}| \leq (3 + \epsilon)^c \cdot a$, for some explicit constant $c$. (For a sufficiently large $a$, the appropriate constant is $c = 4/ \mu + 1$.) Hence, the number of employed colors is $O(a)$. 

Next, we analyze the running time of Procedure Legal-Coloring.
By Lemma \ref{legal1}(2), during the execution of Procedure Legal-Coloring, the Procedure Arbdefective-Coloring is invoked for a constant number of times. Note also that each time it is invoked with the same values of the parameters $t = k = p = \left \lfloor a^{\mu /2} \right \rfloor$. Hence, by Corollary \ref{arb}, executing the while-loop requires $O(t^2 \log n)$ = $O(a^{\mu } \log n)$ time. By Lemma \ref{arbcol}, the additional time required for coloring all the subgraphs in step 19 of Algorithm \ref{proced:legal} is $O(p \log n) = O(a^{\mu/2} \log n)$. (By the termination condition of the while-loop (line 4), once the algorithm reaches line 19, it holds that $\alpha \leq p$.) Therefore, the total running time is $O(a^{\mu} \log n)$.
\end{proof}

Theorem \ref{sumlegal} implies that for the family of graphs with polylogarithmic arboricity in $n$, an $O(a)$-coloring can be computed in time $O((\log n)^{1 + \mu'})$, for an arbitrarily small positive constant $\mu'$. In the case of graphs with superlogarithmic arboricity, we can achieve even better results than those that are given in Theorem \ref{sumlegal}. In this case we execute Procedure Legal-Coloring with the parameter $p = \left \lfloor \frac{a^{\mu'}}{\log n} \right \rfloor$. Since $a$ is superlogarithmic in $n$, and $\mu' > 0 $ is a constant, it holds that $p > a^{\mu' /2}$, for a sufficiently large $n$. Therefore, Procedure Legal-Coloring executes its loop a constant number of times. Consequently, the number of colors employed is still $O(a)$. The running time is the sum of running time of Procedure Arbdefective-Color and the running time of computing legal colorings of graphs of arboricity at most $p$, which is $O(\frac{a^{2 \mu'}}{\log^2 n} \cdot \log n + \frac{a^{\mu'}}{\log n}\cdot {\log n}) = O(a^{2\mu'}).$ If we set $\mu' =  \mu /2$, the running time becomes $O(a^{\mu})$. We summarize this result in the following corollary.

\begin{col}
Let $\mu$ be an arbitrarily small constant. For any graph $G$, a legal $O(a)$-coloring of $G$ can be computed in time $O(a^\mu + (\log n)^{1 + \mu})$.
\end{col}

Next, we demonstrate that one can trade the number of colors for time. Specifically, we show that if one is allowed to use slightly more than $O(a)$ colors, the running time can be bounded by $polylog(n)$, {\em for all values of $a$}. To this end, we select the parameter $p$ to be polylogarithmic in $a$. With this value of $p$ the running time $O(p \log n)$ of the coloring step in line 19 of Algorithm \ref{proced:legal} becomes polylogarithmic. Moreover, setting the parameters $t$ and $k$ to be polylogarithmic in $a$ results in a polylogarithmic running time of Procedure Arbdefective-Coloring. The number of executions of an iteration of the while-loop is $O(\log_p a)$. Consequently, the total running time is also polylogarithmic. However, the number of iterations becomes superconstant. Hence the number of colors grows beyond $O(a)$. The specific parameters we select are $p = k = t = f(a)^{1/2}$, for an arbitrarily slow-growing function $f(a) = \omega(1)$. The results of invoking Procedure Legal-Coloring with these parameters are given below.

\begin{thm} \label{pol}
Invoking Procedure Legal-Coloring with the parameter $p = f(a)^{1/2}$, $f(a) = \omega(1)$ as above, requires  $O(f(a) \log a \log n)$ time. The resulting coloring employs $a^{1 + o(1)}$ colors.
\end{thm}
\begin{proof}
Set $b = p / (3 + \epsilon)$. The number of iterations is at most $\log_b a = O(\frac{\log a}{\log f(a)})$. 
Each iteration requires $O(p^2 \log n) = O(f(a) \log n)$ time. Hence the running time is \ $\log_b a \cdot O(f(a) \log n) = O(f(a) \log a \log n)$.
By Lemma \ref{legal1}(3), the total number of employed colors is at most  $a \cdot (3 + \epsilon)^{O(\log a / \log f(a))}$ 

\noindent $ = a ^{1 + O(1 / \log f(a))} = a^{1 + o(1)}.$
\end{proof}

More generally, as evident from the above analysis, the running time of our algorithm is $O(p^2 \log_p a \log n)$, and the number of colors used is $2^{O(\log_p a)} \cdot a $. Another noticeable point on the tradeoff curve is on the opposite end of the spectrum, i.e., $p = C$, for some large constant $C$. Here the tradeoff gives rise to $a^{1 + O(1/ \log C)}$-coloring in $O(\log a \log n)$ time.

\begin{col} \label{polylog}
For an arbitrarily small constant $\eta > 0$, Procedure Legal-Coloring invoked with $p = 2^{O(1/ \eta)}$ produces an $O(a^{1 + \eta})$-coloring in $O(\log a \log n)$ time.
\end{col}

Corollary \ref{polylog} implies that any graph $G$ for which there exists a constant $\nu > 0$ such that $a \leq \Delta^{1 - \nu}$ can be colored with $o(\Delta)$ colors in $O(\log a \log n)$ time. This goal is achieved by computing an $O(a^{1 + \nu})$-coloring of the input graph $G$. Since 
$a^{1 + \nu} \le \Delta^{1 - \nu^2}$, this is an $o(\Delta)$-coloring of G.
Therefore, our results give rise to determinstic polylogarithmic $(\Delta + 1)$-coloring algorithm for a very wide family of graphs. This fact is summarized in the following corollary.

\begin{col}
For the family of graphs with arboricity $a \leq \Delta^{1 - \nu}$, for an arbitrarily small constant $\nu$, one can compute $(\Delta + 1)$-coloring in $O(\log a \log n)$ time.
\end{col}


 
\section{Even Faster Coloring}

In this section we show that one can decrease the running time of the coloring procedure almost all the way to $\log n$, at the expense of increasing the number of colors. (The number of colors still stays $o(a^2)$, but it grows significantly beyond $a^{1 + \eta}$.) In addition, we show that for any $t$, $1\leq t \leq a$, and any constant $\epsilon > 0$, one can compute $O(a \cdot t)$-coloring in $O((\frac{a}{t})^{\epsilon} \cdot \log n)$ time.

We start with extending an algorithm from \cite{K09} to graphs of bounded arboricity. Specifically, Kuhn \cite{K09} devised an algorithm that works on an $n$-vertex graph $G$ of maximum degree $\Delta$, and for an integer parameter $t$, $1 \leq t \leq \Delta$, it constructs an $O(t^2)$-coloring, $(\Delta/t)$-defective in $O(\log^* n)$ time. (His technique is based on that of Linial \cite{L92}.) We show that if a graph $G$ has arboricity at most $a$ then an $(a/t)$-arbdefective $O(t^2)$-coloring can be computed in $O(\log n)$ time.

The first step of our algorithm is to construct an orientation $\sigma$ of out-degree at most $A$, $A = (2 + \epsilon) \cdot a$, for some constant $\epsilon > 0$. To this end we employ an algorithm from \cite{BE08}. This algorithm requires $O(\log n)$ time, and it is the most time-consuming step of the algorithm. The second step uses this orientation to execute an algorithm that is analogous to the one of \cite{K09}.

Next, we describe this algorithm. Set $d = \Delta/t$. Suppose that we start with a $d'$-arbdefective $M$-coloring of $G$, for some possibly very large $M$, and $0 \leq d' \leq d$. Consider a pair of sets ${\cal A}$ and ${\cal B}$ that will be determined later, and let $F({\cal A}, {\cal B})$ denote the collection of all functions from ${\cal A}$ to ${\cal B}$. Consider also a mapping $\Psi:[M]\rightarrow F({\cal A}, {\cal B})$ that associates a function $\varphi_{\chi} \in F({\cal A}, {\cal B})$ (i.e, $\varphi_{\chi}:{\cal A} \rightarrow {\cal B})$ with each color $\chi \in [M]$. Our coloring algorithm is based on a recoloring subroutine, Procedure Arb-Recolor. This procedure is described in Algorithm \ref{proced:arbrec}. The procedure accepts as input the original color $\chi \in [M]$ of the vertex $v$ that executes the procedure, the $\delta \leq A$ colors of the parents of $v$ according to the orientation computed on the first step of the algorithm, and the defect parameter $d$.

\begin{algorithm}[H]
\caption{Procedure Arb-Recolor}

\label{proced:arbrec}

\textbf{Input:} A color $\chi \in [M]$, parent colors $y_1,y_2,...,y_{\delta} \in [M]$, parameter $d$.

\begin{algorithmic}[1]

\STATE find $\alpha \in {\cal A}$ such that \ \ \  $|\{ i \in [\delta]: \varphi_{\chi}(\alpha) = \varphi_{y_i}(\alpha)\}| \leq d.$ \ \ \  (*)

\STATE return (color := $(\alpha, \varphi_{\chi}(\alpha))$).

\end{algorithmic}
\end{algorithm}

The following lemma (analogous to Lemma 4.1 in \cite{K09}) summarizes the properties of Procedure Arb-Recolor. (Its proof is analogous to that of Lemma 4.1 from \cite{K09}, and it is provided in Appendix B for the sake of completeness.)

\begin{lem} \label {arbkuhn}
For a value $k > 0$, suppose that the functions $\{\varphi_x : x \in [M] \}$ satisfy that for any two distinct colors $x,y \in [M]$, there are at most $k$ values $\alpha \in {\cal A}$ for which $\varphi_x(\alpha) = \varphi_y(\alpha)$. Suppose also that $|{\cal A}| > k \cdot \frac{A - d'}{d - d' + 1}$. Then procedure Arb-Recolor computes a $d$-arbdefective $(|{\cal A}| \cdot | {\cal B}|)$-coloring $\chi'$.
\end{lem}

\def\APPz{
The lemma is proved in two steps. First, we show that if for every vertex $v$ of original color $\chi(v) = \chi$ there exists a value $\alpha$ that satisfies the property (*) of step 1 of Procedure Arb-Recolor, then the arbdefect of the resulting coloring is at most $d$. (Since $\alpha \in {\cal A}$ and $\varphi_x \in {\cal B}$, it is obvious that the resulting coloring employs at most $|{\cal A}| \cdot | {\cal B}|$ colors.) Second, we show that these values $\alpha$ do indeed exist, for all vertices.

Consider a vertex $v$ with $\chi(v) = x$, and with $\delta \leq A$ parents under the orientation $\sigma$  with colors $y_1,y_2,...,y_{\delta} \in [M]$. Suppose that $v$ selects a new color $\chi'(v) = (\alpha, \varphi_x(\alpha))$. Denote $\beta = \varphi_x(\alpha)$. Let $u$ be a parent of $v$ with a color $y \in [M]$ (i.e., $\chi(u) = y$) for which $\beta' = \varphi_y(\alpha) \neq \varphi_x(\alpha) = \beta$. 

Denote by $(\alpha'', \beta'')$ the color selected by $u$ (i.e., $\chi'(u) = (\alpha'', \beta''))$. If $\alpha'' \neq \alpha$ then $\chi'(v) \neq \chi'(u)$. If $\alpha'' = \alpha$ then $\beta'' = \beta' = \varphi_y(\alpha) \neq \beta$, and again $\chi'(v) \neq \chi'(u)$. Hence for $\chi'(u)$ to be equal to $\chi'(v)$ the vertices $u$ and $v$ must select the same value of $\alpha$, and moreover, it must hold that $\varphi_x(\alpha) = \varphi_y(\alpha)$. On the other hand, by (*) of step 2 of Procedure Arb-Recolor, there are at most $d$ indices $i \in [\delta]$ such that $\varphi_x(\alpha) = \varphi_{y_i}(\alpha)$. Hence at most $d$ parents of $v$ may get the same resulting color as $v$.

Now we show that the values $\alpha$ that are required by the algorithm exist. (Observe that their existence is sufficient, since the computations that are needed to find them are local.) Assume without loss of generality that $v$ has exactly $A$ parents, i.e., $\delta = A$. Let $\ell$, $\ell \leq d'$, be the number of parents $u$ of $v$ with the same original color (i.e., $\chi(u) = \chi(v))$. It is sufficient to show that at most $(d - \ell)$ other parents $u$ of $v$ end up to have the same $\chi'$-color as $v$. (As this will imply that the overall number of parents $u$ with $\chi'(u) = \chi'(v)$ is at most $d$.)

Let $S = \{ i \in [\delta]: y_i \neq x \}$ be the set of indices of parents $u$ of $v$ with $\chi(u) \neq \chi(v)$. We argue that there exists an $\alpha \in {\cal A}$ such that $| \{i \in S: \varphi_x(\alpha) = \varphi_{y_i}(\alpha) \} | \leq d - \ell$. For contradiction, suppose that for every $\alpha \in {\cal A}$ there are at least $d- \ell + 1$ indices $i \in S$ for which $\varphi_x(\alpha) = \varphi_{y_i}(\alpha)$. However, for every $i \in S$, $\varphi_x(\alpha) = \varphi_{y_i}(\alpha)$ holds for at most $k$ distinct values of $\alpha \in {\cal A}$. Since there are at most $(A - \ell)$ parents $u$ of $v$ that satisfy $\chi(u) \neq \chi(v) = x$, it follows that $$(A - \ell) \cdot k \geq \Sigma_{i \in S} | \{ \alpha \in {\cal A}: \varphi_x(\alpha) = \varphi_{y_i}(\alpha) \} |.$$
The right-hand side is equal to $\Sigma_{\alpha \in {\cal A}} \left| \{ i \in S: \varphi_x(\alpha) = \varphi_{y_i}(\alpha) \} \right|$. By the contradiction assumption, for every value $\alpha \in {\cal A}$, it holds that $| \{ i \in S: \varphi_x(\alpha) = \varphi_{y_i}(\alpha) \} | \geq d - \ell + 1$.

\noindent Hence $\Sigma_{\alpha \in {\cal A}} \left| \{ i \in S: \varphi_x(\alpha) = \varphi_{y_i}(\alpha) \} \right| \geq (d - \ell + 1) \cdot |{\cal A}|$.

\noindent Hence $(A - \ell) \cdot k \geq (d - \ell + 1) \cdot | {\cal A}|$, i.e., $|{\cal A}| \leq \frac{(A - \ell) \cdot k}{d + 1 - \ell}$.

\noindent For $ A \geq d + 1$, it holds that $\frac{A - \ell}{(d + 1) - d'} \leq \frac{A - d'}{(d + 1) - d'}$, because $\ell \leq d'$. Hence $| {\cal A} | \leq \frac{(A - d')\cdot k }{(d + 1) - d'}$.

\noindent This, however, contradicts the assumption of the lemma about the cardinality of ${\cal A}$. (For $A = d$ a trivial coloring that assigns the same color to all vertices has arbdefect $d$.)

}
By Lemma 4.3 \cite{K09}, the collection of functions $\{ \varphi_x : x \in [M] \}$ with the property required by the statement of Lemma \ref{arbkuhn} exists if $|{\cal B}| \geq \frac{{\cal A}}{2 \cdot \ln M}$ and $k= \left \lfloor 2 e \cdot \ln M \right \rfloor$. For Lemma \ref{arbkuhn} to hold, $|{\cal A}|$ should be greater than $k \cdot \frac{A - d'}{d - d' + 1}$. Hence the number of colors used by $\chi'$ is $|{\cal A}| \cdot |{\cal B}| = k^2 \frac {(A - d')^2}{(d - d' + 1) 2 \ln M} = O(\log M) \cdot \frac{(A - d')^2}{d - d' + 1}$. 

By using $O(\log^* M)$ iterations of Procedure Arb-Recolor with intermediate values of defect parameter that are specified in the proof of Theorem 4.9 of \cite{K09} we obtain an $O(A^2 / d^2)$-coloring with arbdefect at most $d$. (The proof for this statement is identical to the proof of Theorem 4.9 of \cite{K09}.) Henceforth, we refer to this algorithm that invokes Procedure Arb-Recolor $O(\log^* M)$ times by {\em Algorithm Arb-Kuhn}. Since each invocation of procedure Arb-Recolor requires $O(1)$ time, the overall running time of Algorithm Arb-Kuhn is $O(\log^* M) = O(\log^* n)$. (For $M = n$ we start with a trivial legal $n$-coloring that uses each vertex Id as its color.) As $d = A/t$, we obtain an $A/t$-arbdefective  $O(t^2)$-coloring within this running time.

Next, we argue that using Algorithm Arb-Kuhn in conjunction with our algorithm enables one to trade between the running time and number of colors. Specifically, set $d = f(a)$ to be some growing function of the arboricity $a$, i.e., $f(a) = \omega(1)$. Invoke Algorithm Arb-Kuhn. We obtain a decomposition of the original graph $G$ into $O(A^2 / f(a)^2) = O(a^2/ f(a)^2)$ subgraphs of arboricity at most $\alpha = f(a)$ each. Invoke on each of these subgraphs in parallel our algorithm that for $n$-vertex graphs with arboricity $\alpha$ computes an $O(\alpha^{1 + \eta})$-coloring in $O(\log \alpha \log n)$ time, for an arbitrarily small constant $\eta > 0$. Use distinct palettes  of size  $O(\alpha^{1 + \eta})$ for each of the $O(a^2 / \alpha^2)$ subgraphs to get a unified $O(a^2/ \alpha^{1-\eta})$-coloring of the entire graph $G$. The running time of this algorithm is $O(\log^* n + \log \alpha \log n) = O(\log f(a) \cdot \log n)$.
Finally, we set $g(a) = f(a)^{1- \eta}$ and obtain the following Theorem.

\begin{thm}
For an arbitrarily small constant $\eta > 0$, and any function $\omega(1) = g(a) = O(a^{1 - \eta})$, our algorithm computes an $O(a^2 / g(a))$-coloring, within time $O(\log g(a) \cdot \log n)$.
\end{thm}

In particular, by setting $g(a) = 2^{\log^{\zeta} a}$ for some $\zeta > 0$, one can have here an $(a^2/ 2^{\Omega(\log^{\zeta} a)})$-coloring within $O(\log^{\zeta} a \log n)$ time. Also, with $g(a) = \log^c a$, for an arbitrarily large constant $c > 0$, one gets an $O(a^2/ \log^c a)$-coloring in $O(\log \log a \log n)$ time. 

Finally, we show that this technique can be used to obtain a tradeoff between the running time and the number of colors. This new tradeoff improves the previous tradeoff (due to \cite{BE08}) for {\em all}  values of the parameters. Specifically, we have shown that $O(a/t)$-arbdefective $O(t^2)$-coloring can be computed in $O(\log n)$ time. In other words, a graph $G$ of arboricity $a$ can be decomposed into $O(t^2)$ subgraphs of arboricity $\alpha = O(a/t)$ each, in $O(\log n)$ time. By Corollary \ref{sumlegal}, by invoking Procedure Legal-Coloring in parallel on all these subgraphs we obtain an $O(\alpha)$-coloring of each of them. The running time of this step is $O((\frac{a}{t})^{\mu} \cdot \log n)$, for an arbitrarily small constant $\mu > 0$. Using disjoint palettes for each of the subgraphs we merge these colorings into a unified  $O(\alpha \cdot t^2) = O(a \cdot t)$-coloring of the original graph $G$. The last step (the merging) requires no communication. Consequently the total running time of the algorithm is $O((\frac{a}{t})^{\mu} \cdot \log n)$.

\begin{thm}
For any parameter $t$, $1 \leq t \leq a$, and a constant $\mu > 0$, an $O(a/t)$-coloring of a graph of arboricity $a$ can be computed in $O((\frac{a}{t})^{\mu} \cdot \log n)$ time.
\end{thm}

\clearpage
\pagenumbering{roman}
\appendix
\centerline{\LARGE\bf Appendix}
\section{The Length of a Complete Acyclic Orientation}
In this section we show that a graph with a complete acyclic orientation $\sigma$ with length $\ell$ can be legally colored using $\ell + 1$ colors. Consequently, the length of any complete acyclic orientation of a graph $G$ is at least $\chi(G) - 1$. 
To prove this assertion we need the following lemma.
\begin{lem}
For two vertices $v, v'$ with $len_{\sigma}(v) = len_{\sigma}(v')$, there is no edge $(v,v')$ in the graph.
\end{lem}
\begin{proof}
Suppose for contradiction that $e = (v,v') \in E$. Suppose without loss of generality that $(v,v')$ is oriented from $v$ to $v'$. Let $P$ (respectively, $P'$) be the longest path emanating from $v$ (resp., v') with all edges oriented according to $\sigma$. By definition of $len()$, $len(v) = |P| = |P'| = len(v')$. However, the path $\hat P = e \circ P'$ obtained by concatenating the edge $e$ with the path $P'$ is also oriented consistently with the orientation $\sigma$. However, $|\hat P | = |P'| + 1 > |P|$, contradicting the maximality of $P$.
\end{proof}

Consider the following coloring procedure that accepts as input a complete acyclic orientation of length $\ell$. The procedure runs for $\ell + 1$ rounds. In round $i$, for $i = 1,2,...,\ell + 1$, all vertices $v$ for which all the parents of $v$ have already selected a color in round $i-1$ or before, select the color $\varphi(v) = i$, and send a message to all their neighbors. This completes the description of the procedure.

\begin{lem}
$\varphi$ is a legal $(\ell + 1)$-coloring of $G$.
\end{lem}
\begin{proof}
First, observe that each vertex $v$ with $len(v) = i$ selects a color in round $i+1$. Since for all $v \in V$ it holds that $len(v) \leq \ell$, it follows that each vertex of $G$ select a color. For each vertex $v$ it holds that $1 \leq \varphi(v) \leq \ell + 1$. It is left to show that for all edges $(u,v) \in E$ the endpoints $u$ and $v$ select distinct colors. Suppose without loss of generality that $v$ is the parent of $u$ (i.e., the edge is oriented towards $v$). Let $i$ and $j$ be the rounds in which $u$ and $v$ have selected their colors, respectively. Since $v$ is the parent of $u$, it has selected its color before $u$ did. Consequently, $\varphi(u) = i > j = \varphi(v)$.
\end{proof}

\section{Proof of Lemma 5.1:}
\APPz
\hfill \ensuremath{\Box}





\end{document}